\documentclass[copyright,creativecommons]{eptcs}
\usepackage{breakurl}             

\usepackage{amsmath,amsthm,amssymb,xspace,mathpartir}
\usepackage{graphicx}
\usepackage[all]{xy}

\usepackage{tikz}
\usepackage{tikz-cd}
\usetikzlibrary{matrix,arrows}
\usepackage{DotArrow}
\usetikzlibrary{shapes,arrows,automata}
\usetikzlibrary{calc}
\usetikzlibrary{decorations.pathreplacing}
\usepackage{MnSymbol}

\newtheorem{theorem}{Theorem}
\newtheorem{corollary}{Corollary}

\newtheorem{definition}{Definition}
\newtheorem{example}{Example}
\newtheorem{remark}{Remark}

\definecolor{applegreen}{rgb}{0.55, 0.71, 0.0}

\makeatletter
\newcommand{\xleftrightarrow}[2][]{\ext@arrow 3359\leftrightarrowfill@{#1}{#2}}
\makeatother

\newcommand{\xdasharrow}[2][->]{
\tikz[baseline=-\the\dimexpr\fontdimen22\textfont2\relax]{
\node[anchor=south,font=\scriptsize, inner ysep=1.5pt,outer xsep=2.2pt](x){#2};
\draw[shorten <=3.4pt,shorten >=3.4pt,dashed,#1](x.south west)--(x.south east);
}
}

\makeatletter
\def\xtwoheadrightarrowfill@{%
 \arrowfill@\relbar\relbar{\rightarrow\mkern-15mu\rightarrow}}
\newcommand*\xtwoheadrightarrow[2][]{%
 \ext@arrow 0099\xtwoheadrightarrowfill@{#1}{#2}}

\newcommand{\Act}{\ensuremath{\mathit{Act}}}

\newenvironment{todo}{\bigskip\hrule\medskip\noindent}{\medskip\hrule\bigskip}

\newcommand{\dland}{\wedgedot}
\newcommand{\dlandomega}{\,\,{\dland}^{\omega} }
\newcommand{\eUntil}{\wedgedot_]}
\newcommand{\eAfter}{\wedgedot_[}
\newcommand{\eBetween}[1]{{\wedgedot_\less\,}{#1}{\,\wedgedot_\gtr}}
\newcommand{\orLTL}{\! \! \!\shortmid  \! \! \! }

\title{Causality for General LTL-definable Properties}

\author{
Georgiana Caltais
\email{georgiana.caltais@uni-konstanz.de}
\and
Sophie Linnea Guetlein
\email{linnea.guetlein@uni-konstanz.de}
\and
Stefan Leue
\email{stefan.leue@uni-konstanz.de}
\institute{Department for Computer and Information Science
\\University of Konstanz, Germany}
}
\def\titlerunning{Causality for LTL}
\def\authorrunning{G.\ Caltais, S.L. Guetlein \& S.\ Leue}

\begin{document}
\maketitle

\begin{abstract}
In this paper we provide a notion of causality for the violation of general Linear Temporal Logic (LTL) properties. The current work is a natural extension of the 
previously proposed approach handling causality in the context of LTL-definable safety 
properties~\cite{LeiLeu13d,DBLP:conf/vmcai/Leitner-FischerL13}. 
The major difference is that now, counterexamples of general LTL properties are not merely finite traces, but infinite lasso-shaped traces. We analyze such infinite counterexamples and identify the relevant ordered occurrences of causal events, obtained by unfolding the looping part of the lasso shaped counterexample sufficiently many times. The focus is on LTL properties from practical considerations: the current results are to be implemented in QuantUM, a tool for causality checking, that exploits explicit state LTL model checking.
\end{abstract}

\section{Introduction}

The importance and complexity of software driven systems is steadily increasing. Software plays a central r\^ ole in daily used objects, such as computers and mobile phones, but also in other areas, for example medical systems, aircraft and automobiles. Particularly in these latter areas, software failures may entail major environmental harm and/or serious injuries of humans. Software systems whose malfunction has such serious consequences are also called safety-critical systems. This paper addresses methods to analyze models of such systems for the detection of
ordered sequences of events that can be considered causal for the malfunctioning of such a system. Of particular 
importance in this setting is the identification of actual causes, {\it i.e.}, sequences of events that are indispensable 
for the malfunctioning to occur and not just mere ``noise" in the system execution.

Model-checking~\cite{DBLP:books/daglib/0020348} is a formal verification technique for systematically checking whether certain temporal requirements 
are satisfied by a system model. Given a state-based model $M$ of the considered system and a property specification $\varphi$, model checkers return a counterexample if $\varphi$ is not satisfied by $M$. This counterexample typically is an execution trace that includes a violation of the property $\varphi$ and can be used to understand the cause of the property violation and to fix the problem. However, counterexamples can be very long and often contain numerous events that are of no relevance to the violation of $\varphi$. Furthermore, there can be a very large number of counterexample
traces in $M$ that all lead to the violation of $\varphi$. 

In precursory work~\cite{LeiLeu13d,DBLP:conf/vmcai/Leitner-FischerL13},   
model checking of reachability properties has been extended to causality checking by considering all 
traces in a model that violate a system safety property expressed by a reachability property $\varphi$. 
Inspired by the actual cause conditions defined in the Structural Equation Model (SEM) for causality in 
systems~\cite{halpern2005causes,DBLP:conf/ijcai/Halpern15} 
the work in~\cite{LeiLeu13d,DBLP:conf/vmcai/Leitner-FischerL13}
defines actual cause conditions on 
ordered sequences of events that correspond to computations in a Transition System model.

In this paper we provide a notion of causality for the violation of general Linear Temporal Logic (LTL) properties. 
Counterexamples of such properties can always be represented by $\omega$-regular expressions of the 
form $uv^{\omega}$, where $u$ and $v$ are regular expressions~\cite{DBLP:books/sp/Peled2001}. 
Counterexamples of this type are also often referred to as ``lasso-shaped". 
They consist of an initial path fragment that can witness the violation of a ``something bad never happens''-kind of property ({\it i.e.}, a safety property) followed by a loop that can witness the violation of a ``something good eventually happens''-kind of property ({\it i.e.}, a liveness property).
In particular, for the case of ``pure'' safety properties, the lasso ends in a self-loop state where the property is violated.

Consider, for an example, the behaviour of an elevator system as follows:

\begin{example}[Elevator]\label{e.g.:elevator}
The elevator can commute between three floors (0,1,2). On each floor there is a button that can be pressed in order to call the elevator. Whenever a button on some floor is pressed, the elevator will try to go to that floor immediately. If two buttons are pressed, the elevator will go to the lower floor first.
We use the identifier $E_i$ to denote the event ``elevator is on floor $i$'', and $B_i$ to denote ``button on floor $i$ is pressed'', for $i \in \{0, 1, 2\}$.

A liveness property is, for example, that whenever a button on the second floor is pressed, the elevator will go there eventually. 
Assume we are interested in finding those sequences of events that lead to a hazard in which the button on the second floor is pressed, but the elevator never arrives at the second floor.
A corresponding counterexample is a lasso-shaped execution in which $B2$ can be observed, whereas $E2$ occurs neither along the initial path fragment after $B2$, nor in the loop.
\end{example}

In this paper, causality checking for general definable LTL properties is done by identifying the counterexample of a temporal property with a so-called Event Order Logic formula. This formula encodes the relevant ordered occurrences of causal events, and is obtained by unfolding the loop-counterexample ``sufficiently many times".

\paragraph{Related work.}{
The idea of exploiting counterexamples as a debugging aid, in order to understand what determined a certain system failure, has been addressed by other works as well. We refer, for instance, to the results in~\cite{DBLP:conf/cav/BeerBCOT09} that uses the notion of causality in~\cite{halpern2005causes,DBLP:conf/ijcai/Halpern15} and provides the user a visual explanation of the failure by marking causes as red dots along the counterexample trace. In~\cite{DBLP:conf/cav/BeerBCOT09}, causes for the violation of an LTL property are computed via an over-approximation algorithm.
For another example, we refer to the work in~\cite{DBLP:conf/nfm/KumazawaT11}, where errors in system models are elicited from similar counterexamples witnessing the violation of liveness properties. 

It is certain that imposing minimality conditions on the size of causal explanations is desirable.
In our work, we adapt the approach from~\cite{halpern2005causes,DBLP:conf/ijcai/Halpern15}, and formalise a notion of causality that is minimal with respect to the number of events ({\it i.e.}, system actions) it encompasses.
In a similar spirit,~\cite{DBLP:conf/tacas/SchuppanB05} proposes a methodology for computing
shortest counterexamples for symbolic model checking of so-called LTL with past formulae.
Other ``nice to have'' properties of causality such as, compositionality, for instance, were addressed in~\cite{Gossler10,Gossler14,Gossler15,Gossler15b,DBLP:journals/corr/CaltaisLM16}.

For a more comprehensive survey on principles, algorithms and applications of counterexample analysis, we refer to~\cite{DBLP:conf/birthday/ClarkeV03}.
}

\paragraph{Contributions.}{In this paper we introduce a notion of causality with respect to the violation of  general LTL properties.
This is an extension of the work in~\cite{DBLP:conf/vmcai/Leitner-FischerL13}, where causality was handled in the context of safety LTL properties. Our main contributions include an adaptation of the so-called Event Order Logic (EOL) in~\cite{DBLP:conf/vmcai/Leitner-FischerL13} in order to enable compact (finite) descriptions of what caused the violation of a system failure. The proposed notion of causality incorporates a series of properties to be satisfied by the EOL formulae characterising property violations. A notion of soundness and completeness depending on a complete enumeration of all bad traces ({\it i.e.}, counterexamples) and good traces is also established. We show that causality in the sense of~\cite{DBLP:conf/vmcai/Leitner-FischerL13} is equivalent with causality in the current paper, for the case of safety LTL properties.}

\paragraph{Structure of paper.}{Section~\ref{sec:prelim} briefly introduces the formal framework for analysing counterexamples witnessing the violation of LTL system properties in the context of transition system models. The corresponding extension of EOL to describe such counterexamples is provided in Section~\ref{sec:eol}. Section~\ref{sec:cause} introduces the proposed notion of causality. In Section~\ref{sec:sound-complete} we discuss soundness and completeness of our approach. Section~\ref{sec:conc} draws the conclusions and provides pointers to future work.
}

\section{Preliminaries}\label{sec:prelim}
In this section we introduce the formal framework for analysing counterexamples witnessing he violation of Linear Temporal Logic~\cite{pubsdoc:logic-computer-science-second} system properties in the context of transition systems.

\begin{definition}[Transition Systems (TS's)]
A TS is a tuple $T = (S, Act, \rightarrow, I, AP, L)$, where
$S$ is a finite set of states,
$Act$ is a set of actions,
$\rightarrow \subseteq S \times Act \times S$ is a transition relation,
$I \subseteq S$ is a set of initial states,
$AP$ is a set of atomic propositions,
and $L: S \rightarrow \mathcal{P}(AP)$ is a function associating to states in $S$ a set of atomic propositions in $AP$.

For $(s, \alpha, s') \in ~\rightarrow$ we also write $s \xrightarrow{\alpha} s'$.
In the remainder of this paper, for each $\alpha \in Act$, we consider an atomic proposition or \emph{event variable} $a_{\alpha} \in AP$ such that: given $s' \in S$, it holds that $a_{\alpha} \in L(s')$ whenever there exists $s \in S$ with $s \xrightarrow{\alpha} s'$.

We define an execution, or \emph{execution trace} of $T$, as a possibly infinite sequence $\sigma = s_0 \alpha_1 s_1\alpha_2 s_2 \ldots $ with $s_0 \in I$ and $s_i  \xrightarrow{\alpha_{i+1}} s_{i+1}$ for all $i \ge 0$. Moreover, for all $i \ge 0$ we write $\sigma[i\ldots]$ to represent the execution $s_i \alpha_{i+1} s_{i+1} \ldots~$. Additionally, for all $0 \leq i < j$ we write $\sigma[i.. j]$ to represent the finite execution  
$s_i \alpha_{i+1} \ldots \alpha_j s_j $.

%
%
%
\end{definition}

Given an execution trace $\sigma' = s_0' \alpha_1' s_1'\alpha_2' \ldots$, we write $s_0 \alpha_1 : \sigma'$ as a shorthand for the execution trace $\sigma = s_0 \alpha_1 s_0' \alpha_1' s_1'\alpha_2' \ldots$.
For simplicity of notation, we sometimes write $\alpha'_1 \alpha'_2 \ldots$, or $a_{\alpha'_1} a_{\alpha'_2} \ldots$ to equivalently represent $\sigma'$.

\begin{remark} In this paper, we consider TS's without terminal states, i.e., TS's for which all executions are infinite. Observe that this is not a limitation. Finite executions ending in a terminal state $s$ can be straightforwardly extended to infinite executions via a transition $s \xrightarrow{\lambda} s_{\lambda}$ such that $s_{\lambda}$ has a self-loop labelled $\lambda$, i.e., $s_{\lambda} \xrightarrow{\lambda} s_{\lambda}$.
\end{remark}


We further focus on formalising properties of TS's.
A safety property can be seen as a requirement that some bad event never happens. More formally, a property $P_{\textnormal{\it{safe}}}$ is a safety property if and only if every path, or execution that {\it violates} $P_{\textnormal{\it{safe}}}$ has a finite prefix that can not be extended to a path satisfying $P_{\textnormal{\it{safe}}}$. Intuitively, this means that if a safety property is violated, this violation already happens after the model has passed a finite sequence of states and after this finite sequence the violation is unrecoverable. Consequently, if we want to check whether a safety property is satisfied or not, it suffices to only look at finite paths of the system.  
A well known approach for reasoning on the violation of safety properties is model-checking implemented via simple depth first search (DFS) algorithms~\cite{DBLP:books/daglib/0020348}. These algorithms check whether starting from an initial state we can find a path to some state of the model where the bad event specified by the safety property happens. If such a path can be found, the property is violated.

A liveness property requires that some good event eventually happens.
It follows that when checking whether a liveness property is satisfied or not it does {\it not} suffice to only look at finite execution fragments of the system. 
Orthogonally to the model-checking of safety properties, reasoning on the violation of liveness properties is performed via nested depth first search (NDFS)~\cite{DBLP:journals/fmsd/CourcoubetisVWY92}. This algorithm searches for an infinite path in the model, such that the good event described by the liveness property does not hold along that path. If such a path can be found, the property is violated.
%
%

Linear time (LT) properties can be expressed in terms of a safety and a liveness property, based on the {\it Decomposition Theorem} $3.37$ in~\cite{DBLP:books/daglib/0020348}.
Consequently, reasoning on the violation of LT properties requires an NDFS-based approach.

We further provide a brief overview on \emph{Linear Temporal Logic} (LTL)~\cite{pubsdoc:logic-computer-science-second} -- a formalism to describe system properties. 
Intuitively, LTL formulae range over the atomic proposition $\mathit{true}$, that holds in any state of a transition system and, respectively, over atomic propositions $a$ satisfied within a state $s$ whenever the labelling function indicates so. 
Recursively, LTL formulae are defined as \emph{disjunctions} ($\,\,\orLTL\,\,$), \emph{conjunctions} ($\&$) and \emph{negations} ($\sim$) of formulae. The \emph{next} ($X$) operator indicates the satisfiability of a property starting with the ``next'' state, whereas the \emph{until} ($U$) operator indicates the satisfiability of a formula $\phi_1$ all the time until a formula $\phi_2$ is finally satisfied. The \emph{eventually} ($\Diamond$) and \emph{generally} ($\square$) operators indicate the satisfiability of a formula ``at some point'' in the future and, respectively, ``all the time''.

\begin{definition}[Linear Temporal Logic (LTL)]\label{def:LTL-sem}
LTL formulae over the set $AP$ of atomic propositions are built according to the following grammar: 
\[
\phi, \phi_1, \phi_2 ::= {\mathit true} \mid a \mid \phi_1 \orLTL \phi_2\mid \phi_1 ~ \& ~ \phi_2 \mid \sim \phi \mid X \phi \mid \phi_1 ~ U ~ \phi_2 \mid \Diamond \phi \mid \square \phi ~~~~~(a\in AP)
\]

LTL formulae are interpreted over transition systems without terminal states $T = (S, Act, \rightarrow, I, AP, L)$. 
Let $\sigma = s_0 \alpha_1 s_1\alpha_2 s_2 \ldots$ be an execution trace in $T$. The following hold:
\begin{itemize}
\item $\sigma \vDash {\mathit true} $
\item $\sigma\vDash a$ iff $a \in L(s_0)$
\item $\sigma \vDash \sim \phi$ iff not $\sigma  \vDash \phi$
\item $\sigma \vDash \phi_1 \orLTL \phi_2$ iff $\sigma  \vDash \phi_1$ or $\sigma \vDash  \phi_2$
\item $\sigma \vDash \phi_1 ~ \& ~ \phi_2$ iff $\sigma  \vDash \phi_1$ and $\sigma \vDash  \phi_2$
\item $\sigma \vDash X \phi$ iff $\sigma[1..]  \vDash \phi$
\item $\sigma \vDash \phi_1 ~ U ~ \phi_2$ iff $\exists k \ge 0: \sigma[k..] \vDash \phi_2$ and $\forall 0 \le i < k: \sigma[i..] \vDash \phi_1$
\item $\sigma \vDash \Diamond \phi$ iff $\exists k \ge 0: \sigma[k..] \vDash \phi$
\item $\sigma \vDash \square \phi$ iff $\forall k \ge 0: \sigma[k..] \vDash \phi$
\end{itemize}

We write $(\varphi_1 \rightarrow \varphi_2)$ as a syntactic sugar for $(\sim \varphi_1 \orLTL \varphi_2)$, and $(\varphi_1 \leftrightarrow \varphi_2)$ as a syntactic sugar for $((\sim \varphi_1 \rightarrow \varphi_2) \& (\sim \varphi_2 \rightarrow \varphi_1))$.

We say that a transition system $T$ satisfies a LTL formula $\phi$, written as $T \vDash \phi$, if and only if for all executions $\sigma$ of $T$ it holds that $\sigma\vDash \phi$.
\end{definition}

If a system model, or TS in our setting, does not satisfy some given LTL-definable property $P$, there must be an execution witnessing the violation of the property. Such executions are called {\it counterexamples}.
For safety properties, counterexamples are {\it finite} execution fragments that start in an initial state of the system and lead to an undesired state where ``something bad'' actually happens.
For liveness properties, counterexamples must be {\it infinite} executions, because every finite path can still be extended to a path satisfying the liveness property and does, therefore, not suffice as an example for the violation of the property. An infinite path that violates a liveness property is lasso-shaped.

\section{Event Order Logic}\label{sec:eol}

The work in~\cite{DBLP:conf/vmcai/Leitner-FischerL13} introduces the so-called Event Order Logic (EOL). Intuitively, formulae in EOL are used to express causality classes for counterexamples. Causality classes can be seen as generalized counterexamples. A causality class represents several counterexamples, all leading to the property violation in the ``same way". Such counterexamples may only differ in some other events that are not essential for the property violation.

In the case of liveness properties, our counterexamples must be lasso shaped, that is, they contain a loop at the end. The sequence ``$E0, B2, (B1, E1, B0, E0)^{\omega}$" \ is a counterexample of $\varphi := \square (B2 \rightarrow \Diamond E2)$ in the elevator model. In words, the elevator is stuck between the ground floor and the first floor, as $(B1, E1, B0, E0)^{\omega}$ indicates that the sequence $(B1, E1, B0, E0)$ keeps repeating forever. We can say that the cause of the property violation is that the buttons on the ground floor and the first floor are pressed repeatedly and between that, the elevator never has the chance to go to the second floor.

We see that the cause of a liveness property violation must consist of some events happening at the beginning and then some other events happening again and again (in the loop). Hence, in what follows, we propose an extension of EOL in~\cite{DBLP:conf/vmcai/Leitner-FischerL13} with so-called \emph{infinite} formulae to express infinite causal behaviours.

The work in~\cite{DBLP:conf/vmcai/Leitner-FischerL13} introduces two kinds of EOL formulae: \emph{simple} and \emph{complex}.
Intuitively, simple EOL formulae, usually denoted by $\phi$, are built over event variables $a_\alpha$ that are atomic propositions witnessing the execution of actions $\alpha$ \emph{at some point in the future}.
The satisfiability of atomic propositions is different within the frameworks of EOL and LTL: the latter assumes satisfiability with respect to the initial state of a trace, whereas the former has an ``eventually'' component. This difference is formalized in Definition~\ref{def:EOL-sem} providing the semantics of EOL.

Similarly to the case of LTL, simple EOL formulae are inductively defined using \emph{negation} ($\neg$), \emph{conjunction} ($\land$) and \emph{disjunction} ($\lor$). As a consequence of the observation above, formulae of shape $\phi_1 \land \phi_2$ (respectively, $\phi_1 \lor \phi_2$) read as: eventually $\phi_1$ will hold and (respectively, or) eventually $\phi_2$ will hold.

For technical reasons related to the semantics of the aforementioned infinite EOL formulae, we split the complex EOL formulae in~\cite{DBLP:conf/vmcai/Leitner-FischerL13} into: \emph{I-complex} and \emph{G-complex}, respectively.
We refer to Remark~\ref{rm:I-G-complex} for a more detailed explanation.

I-complex formulae, usually denoted by $\psi$, include simple EOL formulae, \emph{conjunctions} ($\land$) and \emph{disjunctions} ($\lor$) of I-complex formulae.
An I-complex formula $\psi_1 \dland \psi_2$ has an ``ordered-and''-like semantics and reads: first $\psi_1$ holds and then $\psi_2$.
Last, but not least, an I-complex formula $\psi_1 \eBetween{\phi} \psi_2$ reads: first $\psi_1$ holds, then $\psi_2$ holds and in between the ``interval'' determined by the satisfiability of $\psi_1$ and $\psi_2$ the simple EOL formula $\phi$ holds all the time.

G-complex formulae, usually denoted by $\theta$, range over I-complex formulae and encompass two more temporal operators: $\eUntil$ that has an ``{until}''-like semantics, and $\eAfter$ that has an ``{after}''-like semantics. More precisely, $\phi \eUntil \theta$ reads: $\theta$ will hold at some point in the future and until then, the simple EOL formula $\phi$ holds all the time.
Orthogonally, $ \theta \eAfter \phi$ reads: at some point $\theta$ holds, and after that, $\phi$ holds all the time.

Observe that the simple and, respectively, G-complex formulae in this paper have the same expressive power as the simple and, respectively, complex EOL formulae originally proposed in~\cite{DBLP:conf/vmcai/Leitner-FischerL13}.

To express infinite causal behaviour, we extend the EOL in~\cite{DBLP:conf/vmcai/Leitner-FischerL13} with the so-called \emph{infinite} formulae, usually denoted by $\xi$.
These are formulae built over the new logical symbol ${\dland}^{\omega}$. For a G-complex formula $\theta$ and an I-complex formula $\psi$ we write $\theta {\dland}^{\omega} \psi$ to express that first $\theta$ holds and then $\psi$ happens infinitely many times.

Formally, as the new EOL we obtain the following:
\begin{definition}[Extended Event Order Logic (EOL) -- Syntax]\label{def:EOL-syntax}
{\emph{Simple EOL formulae}} over a set $\mathcal{A}$ of event variables are formed according to the following grammar:
\[
\phi,\phi_1,\phi_2 ::= \top \mid a_{\alpha} \mid \lnot \phi \mid \phi_1 \land \phi_2  \mid \phi_1 \lor \phi_2~~~~~(a_{\alpha} \in \mathcal{A}).
\]

{\emph{Complex EOL formulae}} are of two kinds:
\begin{itemize}
\item \emph{I-complex EOL formulae}, formed according to the following grammar:
\[
\psi,\psi_1,\psi_2 ::=\phi \mid \psi_1 \dland \psi_2  \mid  \psi_1 \eBetween{\phi}  \psi_2 \mid \psi_1 \land \psi_2 \mid \psi_1 \lor \psi_2 
\]
where $\phi$ is a simple EOL formula.

\item \emph{G-complex EOL formulae}, formed according to the following grammar:
\[
\theta ::= \psi \mid \phi \eUntil \theta \mid \theta \eAfter \phi
\]
where $\phi$ is a simple EOL formula and $\psi$ an I-complex EOL formula.

%
\end{itemize}

{\emph{Infinite EOL formulae}} are formed according to the following grammar:
\[
\xi ::= \theta {\dlandomega} \ \psi
\]
where $\theta$ is a G-complex EOL formula and $\psi$ is an I-complex formula.
\end{definition}

We want an infinite execution $\sigma$ to satisfy an infinite EOL formula $\xi = \theta {\dlandomega} \psi$ if and only if (a) the events in $\theta$ occur in $\sigma$ in the order specified by $\theta$, and (b) the events of $\psi$ occur in $\sigma$ in the order specified by $\psi$, infinitely many times. 

As an example, consider the following execution $\sigma$ in a TS:
\begin{equation}\label{eq:inf-exec1}
\xymatrix@C=1.2cm@R=.25cm{
\sigma = s_0 \ar[r]^{\alpha_1}&  s_1 \ar@/^0.5pc/[r]^{\alpha_2} & s_2 \ar@/^0.5pc/[l]^{\alpha_3}
}
\end{equation}

It is easy to see that $\sigma$ satisfies the formula 
\[
\xi = a_{\alpha_1} \dlandomega ( a_{\alpha_2} \dland  a_{\alpha_3})
\]
where $a_{\alpha_i}$ is the event variable corresponding to $\alpha_i$ for $i \in \{1,2,3\}$.
We see that $\sigma$ contains a finite part $\sigma_1 = s_0 \xrightarrow{\alpha_1} s_1$ determining the event variable $a_{\alpha_1}$ to occur, and a finite part in the loop $\sigma_2=s_1 \xrightarrow{\alpha_2} s_2 \xrightarrow{\alpha_3} s_1$ where $a_{\alpha_2}$ and $a_{\alpha_3}$ occur. Thus, an intuitive approach to decide whether an infinite execution trace satisfies an infinite EOL formula $\xi = \theta {\dlandomega} \psi$ is to split $\sigma$ into an initial final trace $\sigma_1$ and finite trace $\sigma_2$ inside of the loop, and check if $\sigma_1$ satisfies $\theta$ and if $\sigma_2$ satisfies $\psi$, respectively.

The following example shows that it is not enough to split the lasso shaped execution trace $\sigma$ into a first part $\sigma_1$ that contains all states up to the loop and a second part $\sigma_2$ that consists of the loop executed only once.
Consider the execution $\sigma$ in~(\ref{eq:inf-exec1}) and the EOL formula
$
\xi' = a_{\alpha_1} \dland a_{\alpha_3} \dlandomega (a_{\alpha_3} \dland a_{\alpha_2}).
$
Our execution $\sigma$ also satisfies $\xi'$ because in $\sigma$ the event $a_{\alpha_1}$ happens before $a_{\alpha_3}$ and after that, $a_{\alpha_2}$ happens after $a_{\alpha_3}$ infinitely many times.
Hence, in this case, the finite traces guaranteeing the satisfiability of $\xi'$ are as follows: 
\begin{equation}\label{eq:sigma'1}
\sigma_1 = s_0 \xrightarrow{\alpha_1} s_1 \xrightarrow{\alpha_2} s_2 \xrightarrow{\alpha_3} s_1
\end{equation}
is obtained by concatenating the sequence in $\sigma$ up to the loop, with one unfolding of the loop, whereas
\begin{equation}\label{eq:sigma'2}
\sigma_2 = s_1 \xrightarrow{\alpha_2} s_2 \xrightarrow{\alpha_3} s_1 \xrightarrow{\alpha_2} s_2 \xrightarrow{\alpha_3} s_1
\end{equation}
is the unfolding of the loop twice.

By following a similar pattern, consider $\sigma$ in~(\ref{eq:inf-exec1}) and the EOL formula \[
\xi'' = a_{\alpha_1} \dland a_{\alpha_3} \dland a_{\alpha_2} \dland a_{\alpha_3} \dland a_{\alpha_2} \dlandomega (a_{\alpha_2} \dland a_{\alpha_3} \dland a_{\alpha_2} \dland a_{\alpha_2} \dland a_{\alpha_2}).
\]
We want $\sigma$ to satisfy this formula as well, but $\sigma_1$ in~(\ref{eq:sigma'1}) and $\sigma_2$ in~(\ref{eq:sigma'2}) do not satisfy their corresponding EOL formulae in $\xi''$. We have to extend $\sigma_1$ until the loop has been executed three times, while $\sigma_2$ is defined by unfolding the loop four times. Hence, we get:
\begin{equation} 
\begin{array}{rcl}
\sigma_1 & = & s_0 \xrightarrow{\alpha_1} s_1 \xrightarrow{\alpha_2} s_2 \xrightarrow{\alpha_3} s_1 \xrightarrow{\alpha_2} s_2 \xrightarrow{\alpha_3} s_1 \xrightarrow{\alpha_2} s_2 \xrightarrow{\alpha_3} s_1\\
\sigma_2 & = & s_1 \xrightarrow{\alpha_2} s_2 \xrightarrow{\alpha_3} s_1 \xrightarrow{\alpha_2} s_2 \xrightarrow{\alpha_3} s_1 \xrightarrow{\alpha_2} s_2 \xrightarrow{\alpha_3} s_1 \xrightarrow{\alpha_2} s_2 \xrightarrow{\alpha_3} s_1
\end{array}
\end{equation}

Intuitively, $\sigma$ satisfies an infinite EOL formula $\xi = \theta {\dlandomega} \psi$ whenever $\sigma$ can be split into two finite executions $\sigma_1$ and $\sigma_2$, ``large enough'' to satisfy $\theta$ and $\psi$, respectively.


\begin{remark}\label{rm:I-G-complex}
As can be seen in Definition~\ref{def:EOL-syntax}, we choose to classify the complex EOL formulae in~\cite{DBLP:conf/vmcai/Leitner-FischerL13} into I-complex formulae $\psi$ and G-complex formulae $\theta$.
This is because we want to allow only interval-like formulae $\psi$ as right-hand side of the $\dlandomega$ operator.
The occurrence of $\phi \eUntil \psi$ or $\psi \eAfter \phi$ in a cycle does not make sense unless $\psi = \phi$, case in which the corresponding formulae $\xi$ can be equivalently expressed in terms of formulae $\theta \dlandomega (\phi \dland .. \dland \phi)$, where $\phi \dland .. \dland \phi$ stands for finite ordered conjunctions ($\dland$) of simple formulae $\phi$.
\end{remark}


The EOL semantics is translated to the setting of general LTL- properties and infinite loop-traces as follows:
\begin{definition}[EOL -- Semantics]\label{def:EOL-sem}
Let $T = (S, Act, \rightarrow, I, AP, L)$ be a transition system without terminal states. Let $\phi, \phi_1, \phi_2$ be simple EOL formulae, let $\psi, \psi_1, \psi_2$ be complex EOL formulae, let $\theta$ be a G-complex EOL formula, and let $\xi $ be an infinite EOL formula. Let $\mathcal{A}$ be a set of event variables and let $a_{\alpha}, a_{\alpha_i}$ range over arbitrary event variables in $\mathcal{A}$.

The satisfiability of EOL formulae ($\vDash_e$) is defined over execution traces $\sigma = s_0 \alpha_1 s_1 \alpha_2 \ldots$ in $T$.

For \emph{simple EOL} formulae we define:
\begin{itemize}
\item $\sigma \vDash_e \top$, i.e., $\top$ (true) is trivially satisfied by all traces
\item $\sigma \vDash_e a_{\alpha}$ iff $\exists 0 < r: \sigma[0..r] \vDash_e a_{\alpha}$ iff $\exists 0 < j \le r : s_{j-1} \xrightarrow{\alpha} s_{j}$
\item $\sigma \vDash_e \lnot \phi$ iff not $\sigma \vDash_e \phi$
\item $\sigma \vDash_e \phi_1 \land \phi_2$ iff $\exists 0< r:\sigma[0..r]  \vDash_e \phi_1 \land \phi_2$ iff $\exists 0 < r: \sigma[0..r] \vDash_e \phi_1$ and $\sigma[0..r] \vDash_e \phi_2$
\item $\sigma \vDash_e \phi_1 \lor \phi_2$ iff $\exists 0< r:\sigma[0..r]  \vDash_e \phi_1 \lor \phi_2$ iff $\exists 0 < r: \sigma[0..r] \vDash_e \phi_1$ or $\sigma[0..r] \vDash_e \phi_2$
\end{itemize}

For \emph{I-complex EOL} formulae we define:
\begin{itemize}
\item $\sigma \vDash_e \psi_1 \ {\dland} \  \psi_2$ iff $\exists 0  < r : \sigma[0..r] \vDash_e \psi_1 \ {\dland} \  \psi_2$ iff\\
$\exists 0 < j \le k < r : \sigma[0..j] \vDash_e \psi_1$ and $\sigma[k..r] \vDash_e \psi_2$ 
\item $\sigma \vDash_e  \psi_1 \ {\dland}_{<} \ \phi \  {\dland}_{>}  \ \psi_2$ iff $\exists 0  < r : \sigma[0..r] \vDash_e \psi_1 \ {\dland}_{<} \ \phi \  {\dland}_{>}  \ \psi_2$ iff\\
$\exists 0  < j \leq k < r : \sigma[0..j] \vDash_e  \psi_1$ and $\sigma[k..r] \vDash_e  \psi_2$ and $\forall l$ s.t. $j \leq l < k: \sigma[l .. l+1] \vDash_e \phi$
\item $\sigma \vDash_e \psi_1 \land \psi_2$ iff $\exists 0 < r: \sigma[0..r] \vDash_e \psi_1 \land \psi_2$ iff $\exists 0  < r :\sigma[0..r] \vDash_e \psi_1$ and  $\sigma[0..r] \vDash_e \psi_2$
\item $\sigma \vDash_e \psi_1 \lor \psi_2$ iff $\exists 0 < r : \sigma[0..r] \vDash_e \psi_1 \lor \psi_2$ iff $\exists 0  < r :\sigma[0..r] \vDash_e \psi_1$ or  $\sigma[0..r] \vDash_e \psi_2$
\end{itemize}

For \emph{G-complex EOL} formulae we define:
\begin{itemize}
\item $\sigma[i..r] \vDash_e  \phi \  {\dland}_{]} \  \theta$ iff $\exists i \leq j < r : \sigma[j..r] \vDash_e \theta$ and $\forall k$ s.t. $i \leq k < j: \sigma[k..k+1] \vDash_e \phi$
\item $\sigma[i..r] \vDash_e \theta \  {\dland}_{[} \phi$ iff $\exists i < j \leq r: \sigma[i..j] \vDash_e \theta$ and $\forall k$ s.t. $j \leq k < r: \sigma[k..k+1] \vDash_e \phi$
\item $\sigma \vDash_e  \phi \  {\dland}_{]} \  \theta$ iff $\exists 0 \leq r : \sigma[r..] \vDash_e \theta$ and $\forall j$ s.t. $0 \leq j < r: \sigma[j..j+1] \vDash_e \phi$
\item $\sigma \vDash_e \theta \  {\dland}_{[} \phi$ iff $\exists 0  < r : \sigma[0..r] \vDash_e \theta$ and $\forall j$ s.t. $j \geq r: \sigma[j..j+1] \vDash_e \phi$
\end{itemize}

Let $\sigma= s_0 \alpha_1 s_1 \alpha_2 \ldots s_{l} \alpha_{l+1} s_{l+1} \ldots \alpha_{l+m} s_{l+m} \alpha_{l+m+1} s_{l+1} \alpha_{l+2} \ldots$ be an infinite execution trace of $T$ with a loop consisting of $m$ states and starting with $s_{l+1}$. Let $\sigma_2^j = \sigma[l..l+j*m-1]$ be the unfolding of the loop for $j$ times.

For infinite EOL formulae $\xi = \theta \dlandomega \psi$ we define:
\begin{itemize}
\item $\sigma \vDash_e \xi$ iff $ \exists i \geq 0, j \geq 0 : \sigma_1 = \sigma[0..i]$ and $\sigma_2^j = \sigma[l..l+j*m-1]$ and
$\sigma_1 \vDash_e \theta$ and $\sigma_2^j \vDash_e \psi$.
\end{itemize}
\end{definition}


\begin{definition}[EOL Formulae over Executions] 
Let $\sigma= s_0 \alpha_1 s_1 \alpha_2 \ldots s_{l} \alpha_{l+1} \ldots s_{l+m-1} \alpha_{l+m} s_{l} \alpha_{l+1} \ldots$ be an infinite execution trace of $T$ with a loop consisting of $m$ states and starting with $s_{l}$. The EOL formula over $\sigma$ is defined as: $
\xi_{\sigma} := a_{\alpha_1} \dland \ldots \dland a_{\alpha_{l}} \dlandomega (a_{\alpha_{l+1}} \dland \ldots \dland a_{\alpha_{l+m}}).
$
\end{definition}

The following definition will give us a possibility of comparing two EOL formulae. Whenever we have an EOL formula $\xi_1$ and extend it to an EOL formula $\xi_2$ by adding some events to the formula, the set of executions that satisfy the EOL formula $\xi_2$ will be a subset of those executions that satisfy the EOL formula $\xi_1$. Intuitively, this holds as in  $\xi_2$ we have more constraints on the represented execution traces.
Therefore, for two infinite EOL formulae $\xi_1$ and $\xi_2$ we will use the notation  $\xi_1 \subseteq \xi_2$ to express that every execution $\sigma$ that satisfies $\xi_2$ also satisfies $\xi_1$. In that case, $\xi_1$ can be seen as a generalized form of $\xi_2$.
\begin{definition}[EOL Formulae Subset Relationship]\label{def:form-subset-rel} Let $\xi_1 and \xi_2$ be infinite EOL formulae. 
\begin{itemize}
\item $\subseteq: \xi_1 \subseteq \xi_2$ iff  every execution $\sigma$ that satisfies $\xi_2$ also satisfies $\xi_1$. Intuitively, this means that the set of events in $\xi_1$ is a subset of the events in $\xi_2$.
\item $ \subset :\xi_1 \subset \xi_2$ iff  $\xi_1\subseteq \xi_2$  and  $\xi_1 \neq \xi_2$.
\end{itemize}
\end{definition}
As an example we consider the EOL formulae $\xi_1 = E0 \dland  B1 \dland E1$ and $\xi_2 = E0 \dland B1 \dland B2 \dland E1$.
In every execution $\sigma$ that satisfies $\xi_2$, the events $E0$, $B1$ and $E1$ will happen one after the other (but there can be other events happening between them). Therefore, every execution $\sigma$ that satisfies $\xi_2$ also satisfies $\xi_1$. Hence, it holds that $\xi_1 \subseteq \xi_2$.

\section{Causality for general LTL-definable properties}\label{sec:cause}
 
In this section, we formally define the notion of actual causality (AC) for general LTL-definable properties.
The definition follows its counterpart in~\cite{DBLP:conf/vmcai/Leitner-FischerL13}.
The latter is an adoption of the actual causality in~\cite{halpern2005causes}, to the context of concurrent systems.
Next, we provide a brief reminder of the causal setting in~\cite{halpern2005causes}.

In~\cite{halpern2005causes}, systems under analysis are formalized as structural equation models.
Intuitively, structural equations are used to describe causal influence
of variables in the system.
The set of all variables is partitioned into the set $U$ of \emph{exogenous} 
variables that are irrelevant with respect to the causal effect, and the set $V$ of \emph{endogenous} variables that are considered to have a meaningful, 
potentially causal effect. The set $X \subseteq V$ contains all events
that jointly might represent a cause.
A signature $\cal{S}$ is defined as a tuple
$\cal{(U,V,R)}$, where $\cal{U}$ is a finite set of exogenous variables,
$\cal{V}$ is a finite set of endogenous variables, and $\cal{R}$
associates with every variable $Y \in \cal{U} \cup \cal{V}$ 
a nonempty set ${\cal{R}}(Y)$ of possible values for Y. 
A structural equation model over a signature $\cal{S}$ is defined in~\cite{halpern2005causes} as
tuple $M = (\cal{S,F})$, where $\cal{F}$ associates 
with each variable $X \in \cal{V}$ a function denoted
$F_{X}$ that defines the values of all variables in X
given the values of all other variables in $\cal{U} \cup \cal{V}$.
Consider a structural equation model $M = (\cal{S,F})$, a vector 
$\vec{X}$ of variables in $\cal{V}$, and vectors $\vec{x}$ and $\vec{u}$
of values for the variables in $\vec{X}$ and $\cal{U}$.
$M_{\vec{X} \leftarrow \vec{x}}$ denotes the structural equation model for which variables in $\vec{X}$ are set to $\vec{x}$. 
Given a signature $\cal{S = (U, V, R)}$, 
a formula of the form $X = x$, for $X \in \cal{V}$ and $x \in {\cal{R}}(X)$, 
is called a primitive event. A basic causal formula over $\cal{S}$ is one 
of the form $[Y_{1} \leftarrow y_{1}, ... , Y_{k} \leftarrow y_{k}, ]\varphi$
where $\varphi$ stands for the effect, or hazard, and 
$Y_{1}, ... Y_{k}$ and X are variables in $\cal{V}$. 
The formula $[Y_{1} \leftarrow y_{1}, ... , Y_{k} \leftarrow y_{k}, ]\varphi$ 
is abbreviated as $[\vec{Y} \leftarrow \vec{y}]\varphi$.
Intuitively, $[\vec{Y} \leftarrow \vec{y}]\varphi$ states that $\varphi$
holds in a setting in which the values 
of the variables in $\vec{Y}$ are set to the values in $\vec{y}$.
A causal formula $\psi$ is a Boolean combination of basic causal formulae. We write $(M, \vec{u}) \models_{SM} \psi$ whenever
$\psi$ is true in the structural 
model $M$, given the context defined by $\vec{u}$. 
Additionally, $\vec{X} = \vec{x}$ stands for a conjunction of primitive events of the form
$X_{1} = x_{1} \wedge .... \wedge X_{k} = x_{k}$.

An actual cause with respect to the hazard, or effect $\varphi$ is defined in~\cite{halpern2005causes} as follows: 

\begin{definition}[Actual cause~\cite{halpern2005causes}]\label{def:cause-halpern-pearl}
$\vec{X} = \vec{x}$ is an actual cause of $\varphi$ in $(M,\vec{u})$ if the following 
three actual cause conditions (AC) hold:
\begin{description}
\item[AC1:]
$(M,\vec{u}) \models_{SM} (\vec{X} = \vec{x}) \wedge \varphi $.

\item[AC2:]
There exists a partition $(\vec{Z}, \vec{W})$ of $\cal{V}$ 
with $\vec{X} \subseteq \vec{Z}$  and some setting $(\vec{x},\vec{w})$
of the variable in $(\vec{X},\vec{W})$ such that: 
\begin{enumerate}
\item $(M,\vec{u}) \models_{SM} [\vec{X} \leftarrow \vec{x}', \vec{W} \leftarrow \vec{w}']\neg\varphi$

\item $(M,\vec{u}) \models_{SM} [\vec{X} \leftarrow \vec{x}, \vec{W} \leftarrow \vec{w}',  \vec{Z}' \leftarrow \vec{z}^{*}]\varphi$ for all subsets $\vec{Z}'$ of $\vec{Z}$
\end{enumerate}

\item[AC3:] 
$\vec{X}$ is minimal, in the sense that no subset of $\vec{X}$ satisfies conditions 
AC1 and AC2.
\end{description}
\end{definition}

Intuitively, in~\cite{halpern2005causes}, condition AC1 states that there is a setting in which both the cause and the effect occur.
AC2(1) expresses a necessity condition.
It says that for
$\vec{X} = \vec{x}$ to be a cause of $\varphi$, there must be a setting $\vec{x'}$ such that
if $\vec{X}$ is set to $\vec{x}'$, $\varphi$ would not have occurred.
However, as stated in~\cite{halpern2005causes}, 
AC2(1) might be too permissive as
it allows to change the values of the variables in both $X$ and $W$. 
Hence, the change of $\varphi$ form true to false 
could be caused by a change of a variable in $X$ or $W$.
The set $W$ enables expressing so-called ``contingent dependencies''.
For an intuition, consider two events ``Alice presses button B2'' and ``Bob presses button B2'' that enable the elevator to reach the second floor of a building.
We say that the elevator reaching the second floor depends on Alice pressing the button, under the contingency that Bob did not press the button.
AC2(2) constrains AC2(1) by keeping the values
of the variables in $X$ at their original values and only changing the variables 
in $W$. AC2(2) corresponds to a sufficiency condition. Intuitively, setting $\vec{X}$ to $\vec{x}$, guarantees that $\varphi$ holds. The minimality condition in AC3 ensures that only those elements that are essential with respect to $\varphi$ are part of the cause.

Actual causality in the context of TS's and LTL-definable properties is defined as an adoption of~\cite{halpern2005causes} to the setting of concurrent systems, in the spirit of~\cite{DBLP:conf/vmcai/Leitner-FischerL13}. In our setting, $\neg \varphi$ represents the hazard, or the effect.

\begin{definition} [Causality for LTL]\label{def:new-sem}
Let $T = (S, Act,\rightarrow ,I,AP,L)$ be a transition system without terminal states. An EOL formula $\xi$ is considered a cause for the violation of the LTL specifiable property $\varphi$, if the following conditions are satisfied:
 
 \begin{itemize}
\item AC1: There exists an infinite execution $\sigma$ in $T$ such that $\sigma \vDash_e \xi$ and $\sigma \nvDash \varphi$


\item AC2(1): There exists an infinite execution $\sigma$ in $T$ such that $\sigma \nvDash_e \xi$ and $\sigma \vDash \varphi$


\item AC2(2): 
For all infinite executions $\sigma''$ in $T$ with $\sigma'' \vDash_e \xi$ 
it holds that  $\sigma'' \nvDash \varphi$.



\item AC3: The EOL formula $\xi$ is minimal, {\it i.e.}, there does not exist an EOL formula $\xi'$ with $\xi' \subset \xi$ that also satisfies conditions AC1 and AC2. 

\end{itemize}
\end{definition}

AC1 above resembles its counterpart in Definition~\ref{def:cause-halpern-pearl} in the sense that it identifies a setting $\sigma$ that satisfies both the cause $\xi$ and the effect $\varphi$.
AC2(1) entails a necessity condition that identifies a setting witnessed by $\neg\xi$ in which the violation of $\varphi$ would not occur. We fully formalise necessity as a completeness result in Theorem~\ref{thm:comp}, Section~\ref{sec:sound-complete}. AC2(2) is a sufficiency result, stating that satisfying the cause $\xi$ is enough to guarantee the violation of $\varphi$.
AC3 is the minimality condition which states that no true subset of $\xi$ satisfies conditions AC1 and AC2. Intuitively, $\xi$ is in the ``most general form possible''.
Moreover, note that our definition of causality does not employ a notion of ``contingency''. This is because our approach to causality checking is based on a complete exploration of the traces within a TS model and enables the explicit identification of all potential causes.

\begin{example}\label{e.g.:non-occurence} If we want to show that $\xi = E_0 \dland B2 \dlandomega (B1 \dland E1 \dland B0 \dland E0)$ is causal with respect to the violation of the LTL property $\varphi = \square (B2 \rightarrow \Diamond E2)$, we need to show that AC1, AC2 and AC3 are fulfilled for $\xi$.

Consider the infinite execution:
\begin{equation}\label{eq:sigma-eg-cause}
\sigma = E_0 B_2 (B_1 E_1 B_0 E_0)^\omega
\end{equation}


Informally,~(\ref{eq:sigma-eg-cause})  states that when at floor $E_0$, after pressing button $B_2$, only alternations of actions press $B_i$ and reach $E_i$ are possible, where $i \in \{1, 2\}$. Note that $E_2$ is never reached, even if $B_2$ was pressed.

Moreover, consider a behaviour in which the elevator stops at floor $E_2$ infinitely many times after $B_2$ being pressed once:
\begin{equation}\label{eq:sigma-eq-cause2}
\sigma' = E_0 B_2 (B_1 E_1 E_2 B_0 E_0)^\omega
\end{equation} 

At this point, we can infer the following:

\begin{itemize}
\item AC1 is satisfied as, for $\sigma$ in~(\ref{eq:sigma-eg-cause}), $\sigma \vDash_e \xi$ and $\sigma \not \vDash \varphi$ hold.
\item AC2(1) holds for $\sigma'' = E_0 (B_2 E_2 B_1 E_1)^\omega
$, for instance.
\item AC2(2) is not fulfilled as, for instance, for $\sigma'$ in~(\ref{eq:sigma-eq-cause2}), $\sigma' \vDash_e \xi$ and $\sigma' \vDash \varphi$ hold.
\end{itemize}

Hence, $\xi$ is not causal, as it does not prohibit the occurrence of $E_2$.
\end{example}
We observe that in order to compute causes for a property violation according to Definition 8 it is not sufficient to start with an execution trace $\sigma$ that is a counterexample for the property $\varphi$, build the EOL formula $\xi_{\sigma}$ over $\sigma$ and generalize it (in the sense of Definition~\ref{def:form-subset-rel}) until it satisfies conditions AC1-AC3.
Example~\ref{e.g.:non-occurence} shows that also the non-occurrence of events can be causal for the violation of a general LTL-property. 
We further introduce a method to compute the events over $\sigma$ whose non-occurrence is causal for the violation of $\varphi$, and encode this information within the cause $\xi$.

We proceed by first defining a valuation function with respect to a set of event variables $\mathcal{M}$. This function maps an execution trace $\sigma$ to the subset of event variables of $\mathcal{M}$ occurring in $\sigma$.
\begin{definition} [Valuation Function]
Given a transition system $T = (S, Act,\rightarrow ,I,AP,L)$ and a finite set of event variables $\mathcal{M} = \{a_{\alpha_1}, ..., a_{\alpha_n}\}$, we define the valuation function $val_{\mathcal{M}}$ as a function on the set of execution traces of $T$ to the set $\mathcal{P(M)}$ -- the powerset of $\mathcal{M}$.
Let $\sigma$ be an execution trace of $T$. Then:
\[
val_{\mathcal{M}}(\sigma) := \{a_{\alpha} \in \mathcal{M} : \sigma \vDash_e a_{\alpha} \}.
\]
\end{definition}

\begin{definition}[Non-Occurrence of Events]
Let $T = (S, Act, \rightarrow, I, AP, L)$ be a transition system without terminal states, $\varphi$ an LTL-definable property, $\sigma$ an execution trace over $T$ with $\sigma \nvDash\varphi$ and $\xi_{\sigma}$ the EOL formula built over $\sigma$.  Let $\mathcal{A} $ be the set of event variables, let $\mathcal{Z}$ be the set of event variables occurring in $\xi_{\sigma}$ and let $\mathcal{W} := \mathcal{A} \backslash \mathcal{Z}$. We say that that $Q$ is the subset of event variables whose non-occurrence in $\sigma$ is causal for the property violation of $\varphi$, if 
\begin{enumerate}
\item $\xi_{\sigma}$ satisfies AC1 and AC2(1).
\item There exists an execution trace $\sigma''$ with $\sigma'' \vDash_e \xi_{\sigma}$,\\
$ val_{\mathcal{Z}}(\sigma) = val_{\mathcal{Z}}(\sigma'')$, $val_{\mathcal{W}}(\sigma) \neq  val_{\mathcal{W}}(\sigma'')$ and $\sigma'' \vDash \varphi$.
\item $Q \subseteq W$ is the smallest set s.t. for all execution traces $\sigma''$ with $\sigma'' \vDash_e \xi_{\sigma}$ and $ val_{\mathcal{Z}}(\sigma) = val_{\mathcal{Z}}(\sigma'')$ and $val_{Q}(\sigma) =  val_{Q}(\sigma'') = \emptyset$ we have $\sigma'' \nvDash \varphi$.
\end{enumerate}
\end{definition}
As above, let $\sigma$ be a counterexample for the property $\varphi$, let $\xi_{\sigma}$ be the EOL formula built over $\sigma$ and assume  $\xi_{\sigma}$ satisfies the first two conditions of the definition above.

We can now compute the subset $Q$ and determine the location of the event variables $a_{\alpha} \in Q$ in the EOL formula $\xi_{\sigma''}$ built over a $\sigma''$ which we choose as in condition 2 above. We then compare $\xi_{\sigma}$ to $\xi_{\sigma''}$ and prohibit the occurrence of $a_{\alpha} $ in $\xi_{\sigma}$ in the same locations as they occur in $\xi_{\sigma''}$.
%
We repeat the procedure for all $\sigma''$ as in 2 above, and build $\xi_{\sigma}$ in an incremental fashion.
This way we obtain a new EOL formula $\xi_{\sigma}$ that also satisfies condition AC2(2).

If, based on Definition~\ref{def:form-subset-rel},  there is a generalization $\xi'$ of $\xi_{\sigma}$ that also satisfies AC1 and AC2, we replace $\xi_{\sigma}$ by $\xi'$. We repeat this procedure until $\xi_{\sigma}$ is in the most general form possible and, therefore, also satisfies AC3. Since $\xi_{\sigma}$ satisfies AC1-AC3 by construction, it is a cause for the violation of $\varphi$. 

In the context of Example~\ref{e.g.:non-occurence}, from traces $\sigma''$ satisfying condition 2, we obtain intermediate formulae $\xi_{\sigma}$ of shape:
\begin{equation}\label{eq:intermediate-causes}
\begin{array}{rr@{}c@{}l}
\sigma'' = E_0  B_2  E_2 (B_1 E_1 B_0 E_0)^{\omega}&\xi_{\sigma} = (E0 \wedgedot B2 \eAfter \lnot E2) & \dlandomega & (B1 \dland E1 \dland B0 \dland E0)\\
\sigma'' = E_0  B_2 (B_1 E_1 E_2 B_0 E_0)^{\omega} &\xi_{\sigma} = (E0 \wedgedot B2 \eAfter \lnot E2) & \dlandomega & (B1 \dland E1 \eBetween{\lnot E2} B0 \dland E0)\\
\sigma'' = E_0  B_2 (B_1 E_1 B_0  E_2E_0)^{\omega}&\xi_{\sigma} = (E0 \wedgedot B2 \eAfter \lnot E2) & \dlandomega & (B1 \eBetween{\lnot E2}  E1 \eBetween{\lnot E2} B0 \dland E0)\\
\sigma'' = E_0  B_2  (B_1 E_1 E_2 B_0 E_0)^{\omega}&\xi_{\sigma} = (E0 \wedgedot B2 \eAfter \lnot E2) & \dlandomega & (B1 \eBetween{\lnot E2}  E1 \eBetween{\lnot E2} B0 \eBetween{\lnot E2} E0)\\
\ldots
\end{array}
\end{equation}

On top of the formula incrementally derived as in~(\ref{eq:intermediate-causes}), the repeated generalisation procedure entails the EOL formula:
\begin{equation}\label{eq:cause-example2}
\xi_{\sigma} = (B2 \eAfter \lnot E2) \dlandomega (\lnot E2)
\end{equation}
which is semantically equivalent with $B2 \wedgedot \lnot E2$.
Observe that $\xi_{\sigma}$ satisfies AC1 for $\sigma$ in~(\ref{eq:sigma-eg-cause}) and AC2.
$\xi_{\sigma}$ also satisfies AC3, because every superset of $\xi_{\sigma}$ will either violate AC1 or AC2. Thus, the EOL formula $\xi_{\sigma}$ satisfies AC1-AC3 .

Conditions AC1-AC3 do not imply that the order of the occurring events is causal.  Whether the order of events occurring in an EOL formula $\xi$ that satisfies AC1-AC3 is causal or not, can be checked by the following Order Condition (OC). Note that $\xi$ can be causal even if the OC is not satisfied. 



\begin{definition}[Order Condition (OC)]\label{def:OC}
Let $T = ( S, \Act, $ $\rightarrow,$ $ I, AP, L)$ be a transition system without terminal states. 
Let $\sigma$ be an infinite execution trace violating an LTL-definable property $\varphi$. Let $\xi$ be the EOL formula built over $\sigma$. Let $\mathcal{Z}$ be the set of event variables occurring in $\xi$ and let $\mathcal{W} := \mathcal{A} \backslash \mathcal{Z}$.
Consider a set of pairs of event variables over $Y \subseteq \cal{Z}$:
\[
\{(a_{\alpha_i}, a_{\alpha_j}) \mid a_{\alpha_i} \wedgedot a_{\alpha_j} \textnormal{ occurs in } \xi\}.
\]
Let $\xi_{\wedge}$ be the formula obtained by replacing the occurrences
$a_{\alpha_i} \wedgedot a_{\alpha_j}$ in $\psi$ with $a_{\alpha_i} \land a_{\alpha_j}$.

The \emph{order condition} (OC) states that the order of events $a_{\alpha_i}$ and $a_{\alpha_j}$ as above is not causal if the following holds:
$
\textnormal{exists } \sigma' \textnormal{ such that } \sigma' \not \vDash \varphi \textnormal{ and } val_{\cal A}(\sigma) = val_{\cal A}(\sigma') \textnormal{ and } \sigma' \not \vDash_e \xi \textnormal{ and } \sigma' \vDash_e \xi_{\land}.
$
\end{definition}

Note that for the EOL formula $\xi_{\sigma}$ in~(\ref{eq:cause-example2}) the following holds: $\xi_{\sigma} = \xi_{\sigma\land}$. As a consequence, the order of events in $\xi_{\sigma}$ is causal. We say that $\xi_{\sigma}$ satisfies OC.

The following result states the equivalence with the original notion of causality in~\cite{DBLP:conf/vmcai/Leitner-FischerL13}, for the case of safety LTL properties.

\begin{corollary}
Let $T = ( S, \Act, $ $\rightarrow,$ $ I, AP, L)$ be a transition system without terminal states. 
Let $\varphi$ be a safety LTL property. A G-complex EOL formula $\theta$ is a cause in the sense of Definition~\ref{def:new-sem} if and only if it is a cause in the sense of~\cite{DBLP:conf/vmcai/Leitner-FischerL13}.
\end{corollary}
\begin{proof}[Proof Sketch.]
First, recall that counterexamples witnessing the violation of safety properties are finite.
Hence, in~\cite{DBLP:conf/vmcai/Leitner-FischerL13}, the satisfiability of EOL formulae (characterising such counterexamples) was established based on finite traces.
Nevertheless, as can be seen from Definition~\ref{def:EOL-sem}, satisfiability of G-complex formulae can be defined via finite traces as well. In fact, the semantics of G-complex formulae and EOL formulae as in~\cite{DBLP:conf/vmcai/Leitner-FischerL13} coincide.
These being said, the equivalence of the two notions of causality in the context of safety properties
follows immediately by case analysis on AC1--AC3.
AC1 and, respectively, AC3 in Definition~\ref{def:new-sem} are identical to their counterparts in~\cite{DBLP:conf/vmcai/Leitner-FischerL13}.
From AC2(1) and AC2(2) in Definition~\ref{def:new-sem} we can infer AC2(1) in~\cite{DBLP:conf/vmcai/Leitner-FischerL13}. AC2(2) in Definition~\ref{def:new-sem} implies AC2(2) in~\cite{DBLP:conf/vmcai/Leitner-FischerL13}, whereas AC2(1) and AC2(2) in~\cite{DBLP:conf/vmcai/Leitner-FischerL13} imply their counterparts in~Definition~\ref{def:new-sem}.
\end{proof}

We further introduce a definition of causality classes for general LTL-properties. Intuitively, causality classes can be interpreted as ``generalized counterexamples''.

\begin{definition}[Causality Class]
Let  $T = (S, Act,\rightarrow ,I,AP,L)$ be a transition system without terminal states and let $\varphi$ be a general LTL formula. Every infinite EOL formula $\xi =  \theta $ ${\dland}^{\omega}$ $\psi$ that is considered a cause for the violation of $\varphi$, {\it i.e.}, every infinite EOL formula $\xi$ that satisfies AC1-AC3 and OC, defines a {\it causality class} $ CC_{\xi}.$ $CC_{\xi} $ is defined as the set of all valid execution traces in $T$ that satisfy $\xi$.
\end{definition}

For example, the EOL formula $\xi_{\sigma} = (B2 \eAfter \lnot E2) \dlandomega (\lnot E2)$ in~(\ref{eq:cause-example2}) satisfies AC1--AC3 and OC and, therefore, defines a causality class. 
Moreover, note that one execution can belong to more than one causality class.


\section{Completeness and Soundness}\label{sec:sound-complete}

We say that causality checking is \emph{complete} whenever for each possible execution trace that violates an LTL property in the transition system under analysis, there exists a causality class representing this trace.
Therefore, completeness can be seen as a necessity condition.
The completeness of causality checking depends on a complete enumeration of all bad 
and good traces in the system. 

\begin{theorem}[Completeness]\label{thm:comp}
Let $T = ( S, \Act, $ $\rightarrow,$ $ I, AP, L)$ be a transition system without terminal states. 
Let $\sigma$ range over infinite execution traces in $T$ violating an LTL-definable property $\varphi$, {\it i.e.}, $\sigma \nvDash \varphi$.
For each such $\sigma$ there exists a causality class of $\varphi$ containing this trace. 
\end{theorem}
\begin{proof}[Proof]
Let $\varphi$ be a general LTL property and let $\xi =  \xi_1 \lor ... \lor \xi_n$ be the disjunction of all EOL formulae $\xi_i$ that satisfy AC1--AC3 and OC. Let $\sigma$ be a trace such that $\sigma \not \vDash \varphi$. We have to show that $\sigma \in CC_{\xi_i}$ for some $i \in 1, \ldots, n$. Assume that $\sigma$ is not contained in any causality class, that is $\sigma \nvDash_e \xi_i$ for all EOL formulae $\xi_i$ that satisfy the conditions AC1--AC3 and OC.

Let $\xi_{\sigma}$ be the EOL formula representing $\sigma$, and let $\mathcal{Z}$ and $\mathcal{W}$ be the corresponding event variable partitioning, with respect to $\Act$.
Since $\sigma \vDash_e \xi_{\sigma}$ it follows that $\xi_{\sigma}$ is excluded from $\xi$ by one of the AC1--AC3 tests. We will show that this is not the case for any of the conditions AC1--AC3 or OC.
\begin{itemize}
\item AC1 is satisfied because for $\sigma$ it holds that $\sigma \vDash_e \xi_{\sigma}$ and $\sigma \nvDash \varphi$.

\item AC2(1) holds given the assumption that there exist $n$ EOL formulae in $\xi$ that satisfy AC1--AC3 and OC.

\item If AC2(2) fails, then there exists an infinite execution $\sigma''$ with $\sigma'' \vDash_e \xi_{\sigma}$ (and, hence, $ val_{\mathcal{Z}} (\sigma) = val_{\mathcal{Z}} (\sigma'')$) but $\sigma'' \vDash \varphi$.
Let $\xi_{\sigma''}$ be the EOL formula derived from $\sigma''$.
We can now transform $\xi_{\sigma}$ to a new formula ${\xi'_{\sigma}}$ by prohibiting the occurrence of $a_{\alpha} $ in $\xi_{\sigma}$, in the same locations as they occur in $\xi_{\sigma''}$. Consequently, we still have $\sigma \vDash_e {\xi'_{\sigma}}$ and $\sigma \vDash {\xi'_{\sigma}}$ but now,  $\sigma'' \nvDash_e {\xi'_{\sigma}}$. Thus, $\sigma''$ does not influence the satisfaction of AC2(2) by ${\xi'_{\sigma}}$. 
If ${\xi'_{\sigma}}$ still doesn't satisfy AC2(2), {\it i.e.}, if there is another $\sigma'''$ with $\sigma''' \vDash_e \xi $ and $ val_{\mathcal{Z}} (\sigma) = val_{\mathcal{Z}} (\sigma''')$ but $\sigma''' \vDash \varphi$, we repeat the procedure from above. Since the action alphabet is finite and the EOL formulae are finitely representable, this procedure will stop after finitely many steps. The resulted EOL formula ${\xi'_{\sigma}}$ satisfies AC2(2).

\item If AC3 excludes $\xi_{\sigma}$, then there must be some $ {\xi'_{\sigma}}  \subset \xi_{\sigma}$ that satisfies AC1 and AC2. We still have $\sigma \vDash {\xi'_{\sigma}}$ by Definition~\ref{def:new-sem} . 

\end{itemize}
In all cases we obtain an EOL formula ${\xi'_{\sigma}}$ that satisfies AC1-AC3 and OC, such that $\sigma \vDash {\xi'_{\sigma}}$. Thus, $\sigma$ is contained in the causality class $CC_{\xi'_{\sigma}}.$
\end{proof}

We define a causality checking result to be \emph{sound} if whenever the events described by a causality class 
occur, the property violation occurs.

\begin{theorem}[Soundness]\label{them:sound}
Let $T = ( S, \Act, $ $\rightarrow,$ $ I, AP, L)$ be a transition system without terminal states.
Each execution trace $\sigma$ of $T$ contained in a causality class of a general LTL property $\varphi$ is a bad trace, {\it i.e.} $\sigma \nvDash \varphi$. 
\end{theorem}
\begin{proof}[Proof]
Let  $\sigma$ be contained in the causality class $CC_{\xi}$ for some EOL formula $\xi$. Since $\xi$ defines a causality class, it must, by definition, satisfy AC1-AC3 and OC. In particular, $\xi$ must satisfy AC2(2). Since  $\sigma \in  CC_{\xi}$ we have $\sigma \vDash_e \xi$ by definition of causality classes. It follows  from AC2(2) that $\sigma\nvDash \varphi$.
\end{proof}

\section{Conclusions}\label{sec:conc}
We have presented an approach for extending causality checking towards general
LTL-definable properties. To this end, we have reconsidered the actual cause conditions AC1-AC3 and OC,
and adapted them to the lasso-shaped counterexamples that general LTL properties entail. 

For practical reasons related to the implementation of the causality checking procedure, our current results are limited to LTL-definable properties. Nevertheless, in the future, we consider extending the formal framework of causality checking to the more general case of $\omega-$ regular linear-time properties~\cite{DBLP:books/daglib/0020348}.
It should be pointed out that the described adaption can be 
straightforwardly extend to general $\omega$-regular properties, 
corresponding to the expressiveness of B\"uchi automata, 
which are a strictly larger class of properties than LTL~\cite{DBLP:journals/iandc/Wolper83}.

As already mentioned, we consider implementing the current causality checking approach in an automated tool. Of particular interest is {QuantUM}~\cite{DBLP:journals/corr/abs-1107-1198}, a tool that enables the semi-formal specification of systems in terms of SysML~\cite{OMGSysML} and applies LTL model-checking for determining what caused the violation of a safety LTL property.
Recall that our notion of causality relies on the
complete enumeration of system traces. Hence, the main challenge is to determine all (lasso-shaped) counterexamples in an efficient way ({\it e.g.}, on-the-fly~\cite{DBLP:conf/spin/CouvreurDP05,DBLP:conf/tacas/SchwoonE05,DBLP:conf/ppopp/BloemenLP16,DBLP:conf/hvc/BloemenP16}).
Further future research comprises significant case studies in order to asses the scalability of our approach.\\[0.5ex]
{\bf{Acknowledgements. }}{The authors are grateful for the useful comments received from the anonymous reviewers of CREST 2018.
}


\end{document}